\documentclass[superscriptaddress,amsmath,amssymb,11pt,]{revtex4}

\usepackage{amsthm}
\usepackage{charter}
\usepackage{graphicx}
\usepackage{bm}
\usepackage{dsfont}
\usepackage[landscape,papersize={297.1mm,210mm},left=1.9cm,right=1.6cm,top=2.7cm,bottom=2.8cm]{geometry}
\usepackage[pdfstartview=FitH, colorlinks, pdfborder=001, linkcolor=red, anchorcolor=blue, citecolor=blue, urlcolor=blue]{hyperref}

\newtheorem{theorem}{Theorem}

\newtheorem{lemma}{Lemma}

\begin{document}
\title{Monogamy and polygamy for multi-qudit generalized $W$-class
states based on concurrence of assistance and Tsallis-$q$ entanglement of assistance}

\author{Wen Zhou}
\email{2230501027@cnu.edu.cn}
\affiliation{School of Mathematical Sciences, Capital Normal University, Beijing 100048, China}
\author{Zhong-Xi Shen}
\email{18738951378@163.com}
\affiliation{School of Mathematical Sciences, Capital Normal University, Beijing 100048, China}
\author{Dong-Ping Xuan}
\email{2230501014@cnu.edu.cn}
\affiliation{School of Mathematical Sciences, Capital Normal University, Beijing 100048, China}
\author{Zhi-Xi Wang}
\email{wangzhx@cnu.edu.cn}
\affiliation{School of Mathematical Sciences, Capital Normal University, Beijing 100048, China}
\author{Shao-Ming Fei}
\email{feishm@cnu.edu.cn}
\affiliation{School of Mathematical Sciences, Capital Normal University, Beijing 100048, China}

\begin{abstract}
By analyzing the reduced density matrices derived from a generalized $W$-class state under any partition, we present new analytical monogamy inequalities satisfied by the $\alpha$-th ($\alpha\geq\gamma,~\gamma\geq2$) power and $\beta$-th ($0\leq\beta\leq\frac{\gamma}{2},~\gamma\geq2$) power of the concurrence of assistance for multi-qudit generalized $W$-class states, which are demonstrated to be tighter than previous studies through detailed examples. Furthermore, using the Tsallis-$q$ entanglement of assistance, we also establish new monogamy and polygamy relations, which are shown to be valid even for multipartite higher-dimensional states that the CKW inequality is violated.

\end{abstract}

\maketitle

\section{Introduction}
Quantum entanglement~\cite{MAN,RPMK,KSS,HPB,JIV,CYSG} is a fundamental feature of quantum mechanics. It plays a pivotal role in distinguishing the quantum from the classical world.
The principle of entanglement monogamy~\cite{ckw61052306} says that a quantum system, once entangled with one subsystem, will constrain its level of entanglement with the remaining subsystems. The monogamy relations not only dictate the distribution of entanglement within multipartite settings, but also applies to various realms of quantum physics, such as in characterizing the construction of a quantum system~\cite{ZF,RML,ROA} and in proving the nocloning theorem~\cite{M052338}. It also acts as an important gradient in enhancing the security of quantum cryptography~\cite{TFK}.

Mathematically, the monogamy of entanglement is characterized through an entanglement measure $\mathcal{E}$, applied to a tripartite system $\rho_{ABC}$, as follows
\begin{equation}\label{CKW}
\mathcal{E}(\rho_{A|BC})\geq \mathcal{E}(\rho_{AB})+\mathcal{E}(\rho_{AC}),
\end{equation}
which is commonly recognized  as the $\mathrm{CKW}$ inequality \cite{ckw61052306}, where $\rho_{AB}={\rm Tr}_C(\rho_{ABC})$ and $\rho_{AC}={\rm Tr}_B(\rho_{ABC})$ are the reduced density matrices by tracing over the subsystem $C$ and $B$, respectively, with ${\rm Tr}_B$ (${\rm Tr}_C$) denoting the partial trace with respect to the subsystem $B$ ($C$). In Ref.~\cite{OV}, the authors established the longstanding conjecture by Coffman et al., demonstrating that the distribution of bipartite quantum entanglement, quantified by the tangle $\tau$~\cite{ckw61052306}, across $N$ qubits adheres to a tight inequality: $\tau{\rho_{A_1A_2}}+\tau{\rho_{A_1A_3}}+\cdots+\tau{\rho_{A_1A_N}}\leq \tau{\rho_{A_1|A_2\cdots A_N}}$. Here, $\tau{\rho_{A_1|A_2\cdots A_N}}$ represents the bipartite quantum entanglement measured by the tangle when the system is bipartitioned between $A_1$ and the combined subsystem $A_2 A_3\cdots A_N$.
In Ref.~\cite{kw69022309}, the authors introduced a simple identity that effectively delineates the trade-off between entanglement and classical correlation, thereby providing a solid foundation for the derivation of stringent monogamy relations. Furthermore, the authors also established diverse monogamous trade-off relations involving other entanglement and correlation measures. In Ref.~\cite{Xuan131}, the authors studied the quantification of parameterized monogamy relation with equation. The author in Ref.~\cite{GY012405} introduced the concept of complete monogamy and tightly complete monogamy for the $k$-entanglement measures.

Another significant notion is the concept of entanglement of assistance, which has been proved to possess polygamy relation \cite{GBS}. The polygamy relation can be regarded as an additional constraint on entanglement within multipartite quantum systems. Conversely, polygamy relations describe scenarios where entanglement can be shared more freely across multipartite subsystems, providing a complementary perspective to monogamy relations. The duality of entanglement sharability versus entanglement of assistance is evident in the sense that the upper bound for the former is the lower bound for the latter. The polygamy relation for the tripartite state  $\rho_{ABC}$ is defined by:
\begin{equation}
\mathcal{E}_{a}(\rho_{A|BC})\leq\mathcal{E}_{a} (\rho_{AB})+\mathcal{E}_{a}(\rho_{AC}),
\end{equation}
where $\mathcal{E}_{a}$ represents the entanglement measure of assistance that corresponds to $\mathcal{E}$. Furthermore, Guo et al. demonstrates that any entanglement of assistance is polygamous rather than monogamous \cite{GY}. Subsequently, the monogamy and polygamy inequalities have been generalized to $N$-qubit systems~\cite{KIM062328,KIM032335,BXZ,KIM012329}, and this generalization includes diverse bipartite entanglement measures and assisted entanglement measures~\cite{ES,KIM062328,KIM032335}. The authors in Ref.~\cite{LYY385} provide generalized definitions of polygamy in $N$-qubit systems with $N\geq4$, and present polygamy inequalities with a polygamy power.

Recently, Kim et al. demonstrated that the generalised $W$-class (GW) states satisfy the monogamy relations with regard to concurrence~\cite{KS495301}. In Ref.~\cite{ZF}, the authors introduces the general monogamy inequalities for the $\alpha$-th power of concurrence and entanglement of formation for $N$-qubit states. Nonetheless, the monogamy relation (\ref{CKW}) does not consistently apply to $N$-qudit systems. Ou revealed in 2007 that $3\otimes3\otimes3$ systems have counterexamples when considering concurrence as the measure~\cite{O034305}. Furthermore, the authors proved that the monogamy relation does not hold for higher-dimensional systems by using a class of entanglement measures with additivity ~\cite{LMHPAW}. Additionally, a particular entanglement measure known as squashed entanglement has been identified, and it is monogamous for systems of arbitrary dimension \cite{CW}.
%One question that naturally emerges is whether such monogamy relations hold also for higher-dimensional quantum states.

%One way to cope with the problem is to seek and study a class of states in higher dimensional systems satisfying the monogamy relations.

Generally, the concurrence of assistance (CoA) fails to satisfy monogamy relations for general quantum states. As a result, researchers have investigated special classes of quantum states that do satisfy the monogamy relations as defined by CoA. General monogamy inequalities have been formulated for CoA and entanglement of formation in the $N$-qubit GW states~\cite{ZXN}. Subsequently, more tighter monogamous relations had been presented for CoA~\cite{jzx1,jzx2,XHLF}. In Ref.~\cite{SX032344}, the authors presented new class of monogamy inequalities pertaining to the CoA for the reduced density matrices of an $N$-qudit GW states under any partition. Then, in Refs.~\cite{LZZ,XHLF}, the authors further improved the monogamy inequalities derived in Ref.~\cite{SX032344}.

The following sections detail the organization of this article. In Sec~\ref{s1}, we review some related basic concepts. In Sec~\ref{s2} and Sec~\ref{s3}, we present new analytical monogamy inequalities for the $\alpha$-th ($\alpha\geq\gamma,~\gamma\geq2$) and $\beta$-th ($0\leq\beta\leq\frac{\gamma}{2},~\gamma\geq2$) powers of CoA for multi-qudit GW states. These inequalities are tighter than previous corresponding ones, as evidenced by our detailed illustrative examples. In Sec~\ref{s4}, utilizing Tsallis-$q$ entanglement of assistance for the parameter $q$, we also establish new monogamy and polygamy relations for  multi-qudit GW states with $\frac{5-\sqrt{13}}{2}\leq q\leq 2$ and $3\leq q\leq\frac{5+\sqrt{13}}{2}$. Moreover, the recently introduced polygamy inequality has proven to be more capable in addressing the counterexamples raised by the CKW monogamy inequality within higher-dimensional systems. At last, the conclusion and discussion are shown in Sec~\ref{s5}.

\section{Preliminaries}\label{s1}
We initially present the definitions for the pertinent entanglement measures and GW states.
For a bipartite pure state $|\psi\rangle_{AB}$ within the finite-dimensional Hilbert space $\mathcal{H}_A\otimes \mathcal{H}_B$, the concurrence is expressed as~\cite{U032307,RBCHM}
\begin{eqnarray}\label{C1}
\mathcal{C}(|\psi\rangle_{AB})=\sqrt{{2\left[1-\mathrm{Tr}(\rho_A^2)\right]}},
\end{eqnarray}
where $\rho_A=\mathrm{Tr}_B(|\psi\rangle_{AB}\langle\psi|)$ denotes the reduced density matrix. The definition of concurrence for a bipartite mixed state $\rho_{AB}$ is given through its convex roof extension
\begin{eqnarray*}
\mathcal{C}(\rho_{AB})=\min_{\{p_i,|\psi_i\rangle\}}\sum_ip_i\mathcal{C}(|\psi_i\rangle),
\end{eqnarray*}
where the minimum is taken over all possible pure state decompositions  of $\rho_{AB}$ into pure states, represented as $\rho_{AB}=\sum_ip_i|\psi_i\rangle\langle\psi_i|$, with $p_i\geq0$ and $\sum_ip_i=1$. For a 2-qubit mixed state $\rho$, the concurrence of $\rho$ is expressed as \cite{W2245}
\begin{eqnarray*}
\mathcal{C}(\rho)=\max\{\zeta_1-\zeta_2-\zeta_3-\zeta_4,0\},
\end{eqnarray*}
where $\zeta_1\geq\zeta_2\geq\zeta_3\geq\zeta_4$ represents the eigenvalues of the matrix $\sqrt{\sqrt{\rho}\widetilde{\rho}\sqrt{\rho}}$. The expression $\widetilde{\rho}$ is defined as $(\varrho_y\otimes\varrho_y) \rho^\ast (\varrho_y\otimes\varrho_y)$, with $\rho^\ast$ denoting the complex conjugation of $\rho$, and $\varrho_y$ denoting the standard Pauli matrix.

%For a general higher-dimension pure bipartite state $|\Psi\rangle$ in Hilbert space $H_A\otimes H_B$, concurrence is defined as \cite{FMA},
%\begin{eqnarray*}
%C(\Psi)=\sqrt{2[\langle\Psi|\Psi\rangle^{2}-\mathrm{Tr}\rho^{2}_{i}]},
%\end{eqnarray*}
%where $\rho_{i}$ is the reduced density operator obtained by tracing over either subsystems A or B. It is clear that
%$C(\Psi)=0$ if and only if $|\Psi\rangle$ is a product state, i.e. $|\Psi\rangle=|\Psi_{A}\rangle\otimes|\Psi_{B}\rangle$.

%Interestingly, $C(\Psi)$ can be observed through a small number of projective measurement on a twofold copy $|\Psi\rangle\otimes |\Psi\rangle$ of $|\Psi\rangle$ \cite{FMA,FMA1,LF,F,FA}:
%\begin{eqnarray*}
%C(\Psi)=\sqrt{\langle\Psi|\otimes\langle\Psi|~\mathcal{A}~|\Psi\rangle\otimes |\Psi\rangle},~~\mathcal{A}=4P_{-}^{A}\otimes P_{-}^{B},
%\end{eqnarray*}
%where $P_{-}^{A}~(P_{-}^{B})$ is the projector onto the antisymmetric subspace of $H_A\otimes H_A$ ($H_B\otimes H_B$).

The CoA for a tripartite pure state $|\psi\rangle_{ABC}$ is defined as follows~\cite{TFS,YCS},
\begin{eqnarray*}
\mathcal{C}_a(|\psi\rangle_{ABC})\equiv \mathcal{C}_a(\rho_{AB})=\max_{\{p_i,|\psi_i\rangle\}}\sum_ip_i\mathcal{C}(|\psi_i\rangle),
\end{eqnarray*}
where the maximum is obtained by considering all possible decompositions of $\rho_{AB}.$ For a pure state $\rho_{AB}$, it holds that $\mathcal{C}(|\psi\rangle_{AB})=\mathcal{C}_a(\rho_{AB})$.

The $N$-qubit $W$-class states and $N$-qudit GW
states are defined by \cite{KS495301}
\begin{eqnarray}\label{w}
|\mathcal{W}\rangle_{A_{1}A_{2}\cdots A_{N}}=b_{1}|10\cdots0\rangle+b_{2}|01\cdots0\rangle+\cdots+b_{N}|00\cdots1\rangle,
\end{eqnarray}
and
\begin{eqnarray}\label{gw}
|\mathcal{W}^{d}_{N}\rangle_{A_{1}A_{2}\cdots A_{N}}=\sum_{s=1}^{d-1}(b_{1s}|s0\cdots0\rangle+b_{2s}|0s\cdots0\rangle+\cdots+b_{Ns}|00\cdots0s\rangle),
\end{eqnarray}
with the normalizations $\sum_{i=1}^{N}|b_i|^2=1$ and $\sum_{t=1}^{N}\sum_{s=1}^{d-1}|b_{ts}|^{2}=1$, respectively. The states (\ref{gw}) includes the $N$-qubit $W$-class states (\ref{w}) as a special case of $d = 2$.

The formation of a partially coherent superposition $\mathrm{(PCS)}$ between an $N$-qudit $\mathrm{GW}$ state $|\mathcal{W}^{d}_{N}\rangle$ and the $N$-qudit vacuum $|0\rangle^{\otimes N}$ yields a class of higher-dimensional quantum states that are parameterized by two variables \cite{KIM},
\begin{eqnarray*}
\rho^{(q,\lambda)}_{A_{1}A_{2}\cdots A_{N}}
&=&q|\mathcal{W}^{d}_{N}\rangle\langle \mathcal{W}^{d}_{N}|+(1-q)|0\rangle^{\otimes N}\langle0|^{\otimes N}\\
&&+\lambda\sqrt{q(1-q)}\big(|\mathcal{W}^{d}_{N}\rangle\langle0|^{\otimes N}+|0\rangle^{\otimes N}\langle \mathcal{W}^{d}_{N}| \big),
\end{eqnarray*}
where $0\leq q,~\lambda\leq 1.$ For $q = 1$, the $\mathrm{PCS}$ state $\rho^{(q,\lambda)}_{A_{1}A_{2}\cdots A_{N}}$ reduces to the GW state $|\mathcal{W}^{d}_{N}\rangle$. When the coherency $\lambda$ equals 1, the $\mathrm{PCS}$ state $\rho^{(q,\lambda)}_{A_{1}A_{2}\cdots A_{N}}$ reduces to the coherent superposition of the $\mathrm{GW}$ states and the vacuum $\mathrm{(GWV)}$ states,
\begin{eqnarray*}
|\phi\rangle_{A_{1}A_{2}\cdots A_{N}} =\sqrt{q}|\mathcal{W}^{d}_{N}\rangle+\sqrt{1-q}|0\rangle^{\otimes N}.
\end{eqnarray*}
If $\lambda = 0$, the state $\rho^{(q,\lambda)}_{A_{1}A_{2}\cdots A_{N}}$ transforms into a mixed incoherent superposition state,
\begin{eqnarray*}
\rho_{A_{1}A_{2}\cdots A_{N}}= q|\mathcal{W}^{d}_{N}\rangle\langle \mathcal{W}^{d}_{N}|+(1-q)|0\rangle^{\otimes N}\langle0|^{\otimes N}.
\end{eqnarray*}
The $N$-qudit $\mathrm{PCS}$ state's reduced $m$-partite density matrix, denoted as $\rho_{A_{1}A_{2}\cdots A_{m}}$, preserves the characteristics of a $\mathrm{PCS}$ within the confines of the reduced $m$ subsystems~\cite{KIM}.

\section{ Improved monogamy relations for concurrence of assistance}\label{s2}

We introduce a set of tighter monogamy relations expressed through the CoA in this section. Different from the CKW inequality (\ref{CKW}), which is applicable to concurrence, the CoA typically does not satisfy the monogamy relations. Nevertheless, for $N$-qubit $W$ state,
$|\psi\rangle_{AB_1\cdots B_{N-1}}\in \mathcal{H}_{A}\otimes \mathcal{H}_{B_{1}}\otimes\cdots\otimes \mathcal{H}_{B_{N-1}}$, one has \cite{ZXN},
\begin{eqnarray}\label{la2}
\mathcal{C}(\rho_{AB_i})=\mathcal{C}_a(\rho_{AB_i}),
\end{eqnarray}
where $\rho_{AB_i}=\mathrm{Tr}_{B_1\cdots B_{i-1}B_{i+1}\cdots B_{N-1}}(|\psi\rangle_{AB_1\cdots B_{N-1}}\langle\psi|), ~i=1,2,...,N-1$. When $\alpha\geq2$, the CoA $\mathcal{C}_a(|\psi\rangle_{A|B_1\cdots B_{N-1}})$ obeys
the monogamy inequality \cite{ZXN},
\begin{equation}\label{ZXN1}
  \mathcal{C}_a^\alpha(|\psi\rangle_{A|B_1B_2\cdots B_{N-1}})\geq \mathcal{C}_a^\alpha(\rho_{AB_1})+\mathcal{C}_a^\alpha(\rho_{AB_2})+\cdots+\mathcal{C}_a^\alpha(\rho_{AB_{N-1}}).
\end{equation}
The relation (\ref{ZXN1}) has been refined the following form
\begin{eqnarray}\label{jzx2}
\mathcal{C}_a^\alpha(\rho_{A|B_1B_2\cdots B_{N-1}})\geq
&&\mathcal{C}_a^\alpha(\rho_{AB_1})+\frac{\alpha}{2} \mathcal{C}_a^\alpha(\rho_{AB_2})+\cdots+\left(\frac{\alpha}{2}\right)^{z-1}\mathcal{C}_a^\alpha(\rho_{AB_z})\nonumber\\
&&+\left(\frac{\alpha}{2}\right)^{z+1}\left(\mathcal{C}_a^\alpha(\rho_{AB_{z+1}})+\cdots+\mathcal{C}_a^\alpha(\rho_{AB_{N-2}})\right)\nonumber\\
&&+\left(\frac{\alpha}{2}\right)^{z}\mathcal{C}_a^\alpha(\rho_{AB_{N-1}})
\end{eqnarray} in Ref.~\cite{jzx1} and then to be
\begin{eqnarray}\label{JZX1}
\mathcal{C}^{\alpha}_{a}(\rho_{A|B_1 \cdots B_{N-1}})&\geq& \mathcal{C}^{\alpha}_{a}(\rho_{AB_1})+(2^{\frac{\alpha}{2}}-1)\mathcal{C}^{\alpha}_{a}(\rho_{AB_2})+\cdots+(2^{\frac{\alpha}{2}}-1)^{z-1}\mathcal{C}^{\alpha}_{a}(\rho_{AB_z})\nonumber\\
& &+(2^{\frac{\alpha}{2}}-1)^{z+1}(\mathcal{C}^{\alpha}_{a}(\rho_{AB_{z+1}})+\cdots+\mathcal{C}^{\alpha}_{a}(\rho_{AB_{N-2}}))\nonumber\\
& &+(2^{\frac{\alpha}{2}}-1)^z\mathcal{C}^{\alpha}_{a}(\rho_{AB_{N-1}})
\end{eqnarray}
in Ref.~\cite{jzx2}, where $\mathcal{C}(\rho_{AB_i})\geq \mathcal{C}(\rho_{A|B_{i+1}\cdots B_{N-1}})$ for $i=1,2,\cdots,z$, and $\mathcal{C}(\rho_{AB_j})\leq \mathcal{C}(\rho_{A|B_{j+1}\cdots B_{N-1}})$ for $j=z+1,\cdots,N-2$, $1\leq z\leq N-3$ and $N\geq4$.

For any $n$-qudit GW states (\ref{gw}) along with a partition $\mathcal{P}=\{P_1,\ldots,P_N \}$
of the set of subsystems $\mathcal{S}=\{A_{1}, A_2,\ldots,A_{n} \}$, $N\leq n$, one has~\cite{KS495301}
\begin{equation}\label{XD1}
\mathcal{C}^{2}({\rho_{P_t|P_1\cdots\widehat{P}_t\cdots P_N}}) = \sum_{l \neq t}\mathcal{C}^{2}(\rho_{P_t P_l}) = \sum_{l \neq t}[\mathcal{C}_{a}(\rho_{P_t P_l})]^2,
\end{equation}
and
\begin{equation}\label{XD2}
\mathcal{C}(\rho_{P_t P_l})=\mathcal{C}_{a}(\rho_{P_t P_l})
\end{equation}
for all $l \neq t$ and $(P_1\cdots\widehat{P}_t\cdots P_N)=(P_1\cdots{P}_t\cdots P_N)-(P_t)$.
Here $P_{t}\cap P_{l}=\varnothing$ for $t\neq l$, and $\bigcup_{t}P_{t}=\mathcal{S}$.
For the case of $N=3$, we have $\mathcal{C}^{2}_a(\rho_{P_{1}|P_{2}P_{3}})=\mathcal{C}^{2}_a(\rho_{P_{1}P_{2}})+\mathcal{C}^{2}_a(\rho_{P_{1}P_{3}})$.
Consider $\rho_{A_{j_1}A_{j_2}\cdots A_{j_z}}$ as the reduced density matrix of an $n$-qudit GW state (\ref{gw}), with $\{P_1,P_2,P_3\}$ representing a partition of the set $\{A_{j_1},A_{j_2},\cdots,A_{j_z}\},$ where $3\leq z\leq n$. Given that any partition of an $n$-qudit GW state can be regarded as a GW state within  higher-dimensional systems~\cite{KS495301}, we have
\begin{eqnarray}\label{u}
  \mathcal{C}^{\gamma}_a(\rho_{P_{1}|P_{2}P_{3}})
  && = (\mathcal{C}^2_a(\rho_{P_{1}P_{2}})+\mathcal{C}^2_a(\rho_{P_{1}P_{3}})^{\frac{\gamma}{2}}\nonumber\\
  &&=\mathcal{C}^\gamma_a(\rho_{P_{1}P_{2}})\left(1+\frac{\mathcal{C}^2_a(\rho_{P_{1}P_{3}})}{\mathcal{C}^2_a(\rho_{P_{1}P_{2}})}\right)^{\frac{\gamma}{2}}\nonumber\\
  && \geq
  \mathcal{C}^\gamma_a(\rho_{P_{1}P_{2}})+\mathcal{C}^\gamma_a(\rho_{P_{1}P_{3}})
\end{eqnarray}
for $\gamma\geq2$. In the first inequality, we have utilized Eq.(\ref{XD1}). The third inequality arises from the application of the inequality $(1+\theta)^\delta\geq 1+\theta^\delta$, which holds for any real numbers $\delta$ and $\theta$ within the range $0 \leq \theta  \leq 1$ and $0\leq \delta\leq1$.

Denote $\{A,B_1,\cdots,B_{N-1}\}$ a partition of the set $\{A_{j_1},A_{j_2},\cdots,A_{j_t}\},$  $N\leq t\leq n$. In Ref.~\cite{XHLF}, Xie et al. improved the relation (\ref{u}) and established the following monogamy relation for an $N$-qudit GW state by using CoA,
\begin{eqnarray}\label{XHLF1}
\mathcal{C}^{\alpha}_{a}(\rho_{A|B_1\cdots B_{N-1}})&\geq&\mathcal{C}^\alpha_{a}(\rho_{AB_1})+\mathcal{V}_1\mathcal{C}^\alpha_{a}(\rho_{AB_2})+\cdots+\mathcal{V}_1\cdots\mathcal{V}_{z-1}\mathcal{C}^\alpha_{a}(\rho_{AB_z})\nonumber\\
& &+\mathcal{V}_1\cdots\mathcal{V}_{z}(\mathcal{V}_{z+1}\mathcal{C}^\alpha_{a}(\rho_{AB_{z+1}})+\cdots+\mathcal{V}_{N-2}\mathcal{C}^\alpha_{a}(\rho_{AB_{N-2}})\nonumber\\
& &+\mathcal{V}_1\cdots\mathcal{V}_{z}\mathcal{C}^\alpha_{a}(\rho_{AB_{N-1}}),
\end{eqnarray}
where $\alpha\geq\gamma ~(\gamma\geq2)$, with $\mathcal{C}^2_{a}(\rho_{AB_i})\geq l_i\mathcal{C}^2_{a}(\rho_{A|B_{i+1}\cdots B_{N-1}})$, $\mathcal{C}^2_{a}(\rho_{A|B_i\cdots B_{N-1}})\geq \mathcal{C}^2_{a}(\rho_{AB_i})+\omega_i\mathcal{C}^2_{a}(\rho_{A|B_{i+1}\cdots B_{N-1}})$ for $i=1,2,\cdots,z$,  $\mathcal{C}^2_{a}(\rho_{A|B_{j+1}\cdots B_{N-1}})\geq l_j\mathcal{C}^2_{a}(\rho_{AB_j})$, $\mathcal{C}^2_{a}(\rho_{A|B_j\cdots B_{N-1}})\geq\omega_j\mathcal{C}^2_{a}(\rho_{AB_j})+\mathcal{C}^2_{a}(\rho_{A|B_{j+1}\cdots B_{N-1}})$ for $j=z+1,
\cdots,N-2$, $1\leq z\leq N-3$ and $N\geq4$, where $\mathcal{V}_r=(\omega_r+l_r)^{\frac{\alpha}{\gamma}}-l_r^{\frac{\alpha}{\gamma}}$ with $1\le r\le N-2$, $\omega_r\geq1$ and $l_r\geq1$.

In the following, we explore tighter monogamy relations that are fulfilled by the $\alpha$-th ($\alpha\geq\gamma, ~\gamma\geq2$) power of CoA for the $N$-qudit GW states. We first introduce a lemma.

\begin{lemma}\label{L1}
Let $\tau\geq1$ and $\delta\geq1$ be any real numbers. For any $ \theta\geq \tau^{\delta}\geq0$ and non-negative real number $z$, we have for $z\geq 1$,
\begin{equation}\label{aa1}
\begin{aligned}
(1+\theta)^{z}-\theta^{z}\geq(1+\tau^{\delta})^{z}-(\tau^{\delta})^{z}.
\end{aligned}
\end{equation}

\begin{proof}
Consider the function $g(\theta,z)=(1+\theta)^{z}-\theta^{z}$ for $\theta\geq \tau^{\delta}\geq 1$ and $z\geq1$, $\tau\geq1$ and $\delta\geq1$. As $\frac{\partial g(\theta,z)}{\partial \theta}=z[(1+\theta)^{z-1}-\theta^{z-1}]\geq0$, the function $g(\theta,z)$ increases with respect to $\theta$. Since $\theta\geq \tau^{\delta}\geq 1$, $g(\theta,z)\geq g(\tau^{\delta},z)=(1+\tau^{\delta})^{z}-(\tau^{\delta})^{z}$. Thus, we derive the inequality (\ref{aa1}).
\end{proof}
\end{lemma}

In the following, we denote by $\mathcal{C}_{aAB_j}=\mathcal{C}_{a}(\rho_{AB_j})$ for $j=1,2,\cdots,N-1$, and $\mathcal{C}_{aA|B_1B_2\cdots B_{N-1}}=\mathcal{C}_{a}(\rho_{A|B_1 B_2\cdots B_{N-1}})$.

\begin{theorem}\label{coa1}
Let $\rho_{A_{j_{1}}\cdots A_{j_{t}}}$ denote the reduced density  matrix of the $\mathrm{GW}$ state $|\mathcal{W}_{n}^{d}\rangle_{A_{1}\cdots A_{n}}$, and $P=\{A,B,C\}$ a partition of the set
$\{A_{j_{1}}\cdots A_{j_{t}}\}, 3\leq t\leq n$.
Let $\omega$, $l$ and $\delta$ are any real numbers satisfying $\omega\geq1$, $l\geq1$ and $\delta\geq1$ for all $\alpha\geq\gamma,~\gamma\geq2$. We have\\
(1) if $\mathcal{C}_{aAB}^\gamma\geq l^{\delta}\mathcal{C}_{aAC}^\gamma$, the CoA  satisfies
\begin{eqnarray}\label{x1}
\mathcal{C}_{aA|BC}^\alpha\geq \mathcal{C}_{aAB}^\alpha+((\omega+l^{\delta})^\frac{\alpha}{\gamma}
-l^\frac{\delta\alpha}{\gamma})\mathcal{C}_{aAC}^\alpha,
\end{eqnarray}\\
(2) if $\mathcal{C}_{aAC}^\gamma\geq l^{\delta}\mathcal{C}_{aAB}^\gamma$, the CoA  satisfies
\begin{eqnarray}\label{x11}
\mathcal{C}_{aA|BC}^\alpha\geq \mathcal{C}_{aAC}^\alpha+((\omega+l^{\delta})^\frac{\alpha}{\gamma}
-l^\frac{\delta\alpha}{\gamma})\mathcal{C}_{aAB}^\alpha.
\end{eqnarray}
\end{theorem}

\begin{proof}
For any tripartite state $\rho_{ABC}$, according to (\ref{u}),
we derive the relation,
$\mathcal{C}_{aA|BC}^\gamma\geq \mathcal{C}_{aAB}^\gamma+\mathcal{C}_{aAC}^\gamma$ for $\gamma\geq2$. Thus, it follows that there is some $\omega$ with $\omega\geq 1$ for which $\mathcal{C}_{aA|BC}^\gamma\geq \mathcal{C}_{aAB}^\gamma+\omega \mathcal{C}_{aAC}^\gamma$.
If $\mathcal{C}_{aAB}^\gamma\geq l^{\delta}\mathcal{C}_{aAC}^\gamma$, we get
\begin{eqnarray*}
\mathcal{C}_{aA|BC}^\alpha&=&(\mathcal{C}_{aA|BC}^\gamma)^{\frac{\alpha}{\gamma}}\geq(\mathcal{C}_{aAB}^\gamma+\omega \mathcal{C}_{aAC}^\gamma)^{\frac{\alpha}{\gamma}}\nonumber\\
&=&\omega^{\frac{\alpha}{\gamma}}\mathcal{C}_{aAC}^{\alpha}[(\omega^{-1}({\mathcal{C}_{aAB}^\gamma}/{\mathcal{C}_{aAC}^\gamma})+1)^{\frac{\alpha}{\gamma}}
-(\omega^{-1}({\mathcal{C}_{aAB}^\gamma}/{\mathcal{C}_{aAC}^\gamma}))^{\frac{\alpha}{\gamma}}]+\mathcal{C}_{aAB}^{\alpha}\nonumber\\
&\geq&\omega^\frac{\alpha}{\gamma}\mathcal{C}_{aAC}^{\alpha}[(\omega^{-1}l^{\delta}+1)^{\frac{\alpha}{\gamma}}-(\omega^{-1}l^{\delta})^{\frac{\alpha}{\gamma}}]
+\mathcal{C}_{aAB}^{\alpha}\nonumber\\
&=&\mathcal{C}_{aAB}^\alpha+((\omega+l^{\delta})^\frac{\alpha}{\gamma}
-l^\frac{\delta\alpha}{\gamma})\mathcal{C}_{aAC}^\alpha,
\end{eqnarray*}
where the second inequality follows from Lemma \ref{L1}. If $\mathcal{C}_{aAB}=0$, then $\mathcal{C}_{aAC}=0$, hence the lower bound is zero trivially. Thus, we get the inequality (\ref{x1}). If $\mathcal{C}_{aAC}^\gamma\geq l^{\delta}\mathcal{C}_{aAB}^\gamma$,  analogously we have $\mathcal{C}_{aA|BC}^\gamma\geq \mathcal{C}_{aAC}^\gamma+\omega \mathcal{C}_{aAB}^\gamma$, which gives rise to (\ref{x11}).
\end{proof}

For multi-qudit systems, the following theorem holds.

\begin{theorem}\label{coaa2}
Let $\rho_{A_{j_1}A_{j_2}\cdots A_{j_t}}$ denote the reduced density matrix of an $n$-qudit GW state $|\mathcal{W}_{n}^{d}\rangle_{A_{1}\cdots A_{n}}$ and $\{A,B_1,\cdots,B_{N-1}\}$ a partition of the set $\{A_{j_1},A_{j_2},\cdots,A_{j_t}\},$ $N\leq t\leq n$.
Consider the real numbers $\delta_{r}\geq1$, $\omega_r\geq1$, $\gamma\geq2$ and $l_r\geq1$, where $1\le r\le N-2$.
If $\mathcal{C}_{aAB_i}^\gamma\geq l_{i}^{\delta_{i}}\mathcal{C}_{aA|B_{i+1}\cdots B_{N-1}}^\gamma$, $\mathcal{C}_{aA|B_i\cdots B_{N-1}}^\gamma\geq \mathcal{C}_{aAB_i}^\gamma+\omega_i\mathcal{C}_{aA|B_{i+1}\cdots B_{N-1}}^\gamma$ for $i=1,2,\cdots,z$, and $\mathcal{C}_{aA|B_{j+1}\cdots B_{N-1}}^\gamma\geq l_{j}^{\delta_{j}}\mathcal{C}_{aAB_j}^\gamma$, $\mathcal{C}_{aA|B_j\cdots B_{N-1}}^\gamma\geq\omega_j\mathcal{C}_{aAB_j}^\gamma+\mathcal{C}_{aA|B_{j+1}\cdots B_{N-1}}^\gamma$ for $j=z+1,
\cdots,N-2$, $1\leq z\leq N-3$ and $N\geq4$, we get
\begin{eqnarray}
\mathcal{C}_{aA|B_1\cdots B_{N-1}}^\alpha&\geq&\mathcal{C}_{aAB_1}^\alpha+
\sum\limits_{i=2}^{z}\prod\limits_{s=1}^{i-1}\Omega_{s}
\mathcal{C}_{aAB_{i}}^{\alpha}\nonumber\\
& &+\Omega_1\cdots\Omega_{z}
\Big(\sum\limits_{j=z+1}^{N-2}
\Omega_{j}\mathcal{C}_{aAB_{j}}^{\alpha}\Big)\nonumber\\
& &+\Omega_1\cdots\Omega_{z}\mathcal{C}_{aAB_{N-1}}^\alpha
\end{eqnarray}
for all $\alpha\geq\gamma,~\gamma\geq2$, where $\Omega_{r}=(\omega_{r}+l_{r}^{\delta_{r}})^\frac{\alpha}{\gamma}
-l_{r}^\frac{\delta_{r}\alpha}{\gamma}$.
\end{theorem}

\begin{proof}
According to Theorem \ref{coa1}, if $\mathcal{C}_{aAB_i}^\gamma\geq l_{i}^{\delta_{i}}\mathcal{C}_{aA|B_{i+1}\cdots B_{N-1}}^\gamma$, $\mathcal{C}_{aA|B_i\cdots B_{N-1}}^\gamma\geq \mathcal{C}_{aAB_i}^\gamma+\omega_i\mathcal{C}_{aA|B_{i+1}\cdots B_{N-1}}^\gamma$ with $i=1,2,\cdots,z$, we have
\begin{eqnarray}\label{coa11}
\mathcal{C}_{aA|B_1\cdots B_{N-1}}^\alpha&\geq&\mathcal{C}_{aAB_1}^\alpha+\Omega_1\mathcal{C}_{aA|B_2\cdots B_{N-1}}^\alpha\nonumber\\
&\geq&\mathcal{C}_{aAB_1}^\alpha+\Omega_1\mathcal{C}_{aAB_2}^\alpha+\Omega_1\Omega_{2}\mathcal{C}_{aA|B_3\cdots B_{N-1}}^\alpha\geq\cdots\nonumber\\
&\geq&\mathcal{C}_{aAB_1}^\alpha+\Omega_1\mathcal{C}_{aAB_2}^\alpha+\cdots+\Omega_1\cdots\Omega_{z-1}\mathcal{C}_{aAB_z}^\alpha\nonumber\\
& &+\Omega_1\cdots\Omega_{z}\mathcal{C}_{aA|B_{z+1}\cdots B_{N-1}}^\alpha.
\end{eqnarray}
If $\mathcal{C}_{aA|B_{j+1}\cdots B_{N-1}}^\gamma\geq l_{j}^{\delta_{j}}\mathcal{C}_{aAB_j}^\gamma$, $\mathcal{C}_{aA|B_j\cdots B_{N-1}}^\gamma\geq\omega_j\mathcal{C}_{aAB_j}^\gamma+\mathcal{C}_{aA|B_{j+1}\cdots B_{N-1}}^\gamma$ with $j=z+1,
\cdots,N-2$, we have
\begin{eqnarray}\label{coa22}
\mathcal{C}_{aA|B_{z+1}\cdots B_{N-1}}^\alpha&\geq&\Omega_{z+1}\mathcal{C}_{aAB_{z+1}}^\alpha+\mathcal{C}_{aA|B_{z+2}\cdots B_{N-1}}^\alpha\nonumber\\
&\geq&\Omega_{z+1}\mathcal{C}_{aAB_{z+1}}^\alpha+\Omega_{z+2}\mathcal{C}_{aAB_{z+2}}^\alpha+\mathcal{C}_{aA|B_{z+3}\cdots B_{N-1}}^\alpha\geq\cdots\nonumber\\
&\geq&\Omega_{z+1}\mathcal{C}_{aAB_{z+1}}^\alpha+\Omega_{z+2}\mathcal{C}_{aAB_{z+2}}^\alpha+\cdots+\Omega_{N-2}\mathcal{C}_{aAB_{N-2}}^\alpha+\mathcal{C}_{aAB_{N-1}}^\alpha.
\end{eqnarray}
Through the conjunction of inequalities (\ref{coa11}) and (\ref{coa22}), we conclude the proof.
\end{proof}

\noindent{\bf Remark 1} \, \
We have established a class of lower bounds for the $\alpha$-th ($\alpha\geq\gamma,~\gamma\geq2$) of CoA, utilizing various parameters $\delta_{r}$, $\omega_{r}$ and $l_{r}$  for the $N$-qudit GW states, where $~\omega_{r}\geq1,~l_{r}\geq1,~\delta_{r}\geq1$, $1\leq r\leq N-2$. When $\delta_{r}=1$, Theorem \ref{coaa2} reduces to the result (\ref{XHLF1}) given in \cite{XHLF}. Set $\gamma=2$ and $d=2$. When $\delta_{r}=1,~\omega_{r}=1,~l_{r}=1$, it is clear that the inequality (\ref{JZX1}) presented in \cite{jzx2} is actually a special instance of Theorem 2 that we have presented.
Additionally, the lower bound for the $\alpha$-th power of CoA becomes tighter as $l_{r}$, $\delta_{r}$ (or $\omega_{r}$) increases, respectively. Since
\begin{equation}
(\omega_{r}+l_{r}^{\delta_{r}})^{\frac{\alpha}{\gamma}}-(l_{r}^{\delta_{r}})^{\frac{\alpha}{\gamma}}
\geq(\omega_{r}+l_{r})^{\frac{\alpha}{\gamma}}-l_{r}^{\frac{\alpha}{\gamma}}
\geq 2^{\frac{\alpha}{\gamma}}-1,
\end{equation}
when $\delta_{r}=1$, the first equality holds, and with $\omega_{r}=1,~l_{r}=1$, the second equality is observed, our monogamy inequalities are tighter than the ones established in Refs.~\cite{XHLF,jzx1,jzx2,ZXN}. To exemplify these advantages, we provide a comprehensive example below.

\noindent{\bf Example 1} \, \ Let us consider the following four-qubit GW state,
\begin{equation}\label{Con6}
|\psi\rangle_{A_{1}A_{2}A_{3}A_{4}}=0.3|0001\rangle+0.4|0010\rangle+0.5|0100\rangle
+\sqrt{0.5}|1000\rangle.
\end{equation}
Tracing over the subsystem $A_4$ we have $\rho_{A_{1}A_{2}A_{3}}=\rho_{ABC}=0.09|000\rangle\langle000|+(0.4|001\rangle+0.5|010\rangle
+\sqrt{0.5}|100\rangle)
(0.4\langle001|+0.5\langle010|+\sqrt{0.5}\langle100|)$. By simple calculation, we obtain $\mathcal{C}_{aA|BC}=\sqrt{\frac{41}{50}}$, $\mathcal{C}_{aAB}=\frac{\sqrt{2}}{2}$ and
$\mathcal{C}_{aAC}=\frac{2\sqrt{2}}{5}$.
Since $\mathcal{C}_{aAB}^\gamma\geq l^{\delta}\mathcal{C}_{aAC}^\gamma$, we have $l^{\delta}\leq \frac{25}{16}$ for $\gamma=2$. Setting $\delta=1.3$, we get $1\leq l\leq1.40959$.
Let $h=1.3$ and $\mu=1$. Then $\mathcal{C}_{aAB}^\alpha+\mathcal{C}_{aAC}^\alpha
=\big(\frac{\sqrt{2}}{2}\big)^\alpha+\big(\frac{2\sqrt{2}}{5}\big)^\alpha$ from Eq.(\ref{ZXN1}) in Ref.~\cite{ZXN},
$\mathcal{C}_{aAB}^\alpha+\frac{\alpha}{2}\mathcal{C}_{aAC}^\alpha=\big(\frac{\sqrt{2}}{2}\big)^\alpha+\frac{\alpha}{2}\big(\frac{2\sqrt{2}}{5}\big)^\alpha$
from Eq.(\ref{jzx2}) in Ref.~\cite{jzx1},
$\mathcal{C}_{aAB}^\alpha+\big(2^{\frac{\alpha}{2}}-1\big)\mathcal{C}_{aAC}^\alpha=\big(\frac{\sqrt{2}}{2}\big)^\alpha+
\big(2^{\frac{\alpha}{2}}-1\big)\big(\frac{2\sqrt{2}}{5}\big)^\alpha$ from Eq.(\ref{JZX1}) in Ref.~\cite{jzx2},
$\mathcal{C}_{aAB}^\alpha+[(\omega+l)^{\frac{\alpha}{2}}-h^{\frac{\alpha}{2}}]\mathcal{C}_{aAC}^\alpha=\big(\frac{\sqrt{2}}{2}\big)^\alpha+[(1+1.3)^{\frac{\alpha}{2}}-1.3^{\frac{\alpha}{2}}]\big(\frac{2\sqrt{2}}{5}\big)^\alpha$
from Eq.(\ref{XHLF1}) in Ref.~\cite{XHLF} and $\mathcal{C}_{aAB}^\alpha+[(\omega+l^{\delta})^{\frac{\alpha}{2}}-(h^{\delta})^{\frac{\alpha}{2}}]\mathcal{C}_{aAC}^\alpha=\big(\frac{\sqrt{2}}{2}\big)^\alpha+[(1+1.3^{\delta})^{\frac{\alpha}{2}}-1.3^{\frac{\delta\alpha}{2}}]\big(\frac{2\sqrt{2}}{5}\big)^\alpha$
from our result (\ref{x1}). It is evident that our result (\ref{x1}) is better than the result (\ref{XHLF1}) from Ref.~\cite{XHLF} when $\alpha\geq2$, hence better than Eq.(\ref{ZXN1}), Eq.(\ref{jzx2}) and Eq.(\ref{JZX1}) given in Ref.~\cite{ZXN} , Ref.~\cite{jzx1} and Ref.~\cite{jzx2}, respectively, see Fig.\ref{Fig1}.
\begin{figure}
  \centering
  \includegraphics[width=10cm]{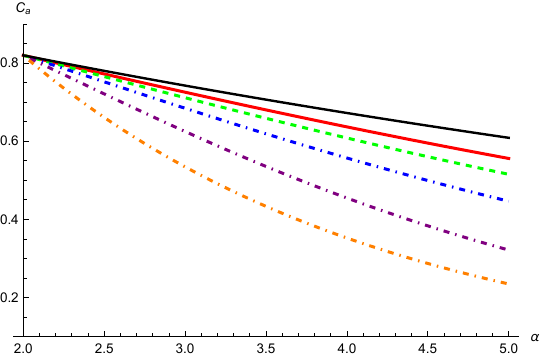}
  \caption{The $y$-axis represents the lower bound for the concurrence of assistance $\mathcal{C}^\alpha_{aA|BC}$.
  Solid black line denotes $\mathcal{C}^\alpha_{aA|BC}$ for the state given in Eq.(\ref{Con6}). The red thick (green dashed, blue dot dashed thick, purple dot dashed and orange dot dashed) line corresponds to the lower bound derived from our result (\ref{x1}) (Ref.~\cite{XHLF}, Ref.~\cite{jzx2}, Ref.~\cite{jzx1} and Ref.~\cite{ZXN}) respectively. }
  \label{Fig1}
\end{figure}

\section{Tighter complementary monogamy relations for concurrence of assistance}\label{s3}
In this section, we present  new analytical monogamy inequalities for the $\beta$-th~ ($0\leq\beta\leq\frac{\gamma}{2},~\gamma\geq2$) power of the CoA for multi-qudit GW states. These inequalities are shown to be tighter than those from previous studies, as evidenced by our detailed examples. We refer to them as complementary monogamy relations, complementing the $\alpha$-th~($\alpha\geq\gamma,~\gamma\geq2$)  power relations from the previous section.

Let $\rho_{A_{j_{1}}\cdots A_{j_{t}}}$ denote the reduced density  matrix of the $\mathrm{GW}$ state $|\mathcal{W}_{n}^{d}\rangle_{A_{1}\cdots A_{n}}$ and $\{A,B_1,\cdots,B_{N-1}\}$ a partition of the set $\{A_{j_1},A_{j_2},\cdots,A_{j_t}\},$  $N\leq t\leq n$.
If $\mathcal{C}_{a}(\rho_{AB_i})\leq \mathcal{C}_{a}(\rho_{A|B_{i+1}\cdots B_{N-2}})$ for $i=1, 2, \cdots, z$, and
$\mathcal{C}_{a}(\rho_{AB_j})\geq \mathcal{C}_{a}(\rho_{A|B_{j+1}\cdots B_{N-1}})$ for $j=z+1,\cdots,N-2$, $0\leq\beta\leq\gamma,~\gamma\geq2$, then the CoA $\mathcal{C}_a(\rho_{A|B_1\cdots B_{N-1}})$ fulfills
the monogamy inequality \cite{SX032344},
\begin{equation}\label{SX}
  \mathcal{C}_a^\beta(\rho_{A|B_1B_2\cdots B_{N-1}})\geq
  \sum^{z}_{i=1}l^{i}\mathcal{C}_a^\beta(\rho_{AB_i})+l^{z}\sum^{N-2}_{i=z+1}\mathcal{C}_a^\beta(\rho_{AB_i})
  +l^{N-1}\mathcal{C}_a^\beta(\rho_{AB_{N-1}}),
\end{equation}
 where $l=2^{\frac{\beta}{\gamma}}-1$.
In Ref.~\cite{LZZ} the relation (\ref{SX}) has been refined the following form
\begin{equation}\label{LYY}
  \mathcal{C}_a^\beta(\rho_{A|B_1B_2\cdots B_{N-1}})\geq
  \sum^{z}_{i=1}l^{i-1}\mathcal{C}_a^\beta(\rho_{AB_i})+l^{z}\sum^{N-2}_{i=z+1}\mathcal{C}_a^\beta(\rho_{AB_i})
  +l^{z+1}\mathcal{C}_a^\beta(\rho_{AB_{N-1}})
\end{equation}
for all $0\leq\beta\leq\gamma,~\gamma\geq2$, where $l=\frac{(1+k)^{\frac{\beta}{\gamma}}-1}{k^{\frac{\beta}{\gamma}}}$,
if $l\mathcal{C}^{\alpha}_a(\rho_{AB_i})\leq \mathcal{C}^{\alpha}_a(\rho_{A|B_{i+1}\cdots B_{N-2}})$ for $i=1, 2, \cdots, z$, and
$\mathcal{C}^{\alpha}_a(\rho_{AB_j})\geq l\mathcal{C}^{\alpha}_a(\rho_{A|B_{j+1}\cdots B_{N-1}})$ for $j=z+1,\cdots,N-2$.

The relation (\ref{LYY}) has been further refined to be
\begin{eqnarray}\label{XHLF2}
\mathcal{C}^{\beta}_{a}(\rho_{A|B_1\cdots B_{N-1}})&\geq&\mathcal{C}^\beta_{a}(\rho_{AB_1})+\mathcal{V}_1\mathcal{C}^\beta_{a}(\rho_{AB_2})+\cdots+\mathcal{V}_1\cdots\mathcal{V}_{z-1}\mathcal{C}^\beta_{a}(\rho_{AB_z})\nonumber\\
& &+\mathcal{V}_1\cdots\mathcal{V}_{z}(\mathcal{V}_{z+1}\mathcal{C}^\beta_{a}(\rho_{AB_{z+1}})+\cdots+\mathcal{V}_{N-2}\mathcal{C}^\beta_{a}(\rho_{AB_{N-2}})\nonumber\\
& &+\mathcal{V}_1\cdots\mathcal{V}_{z}\mathcal{C}^\beta_{a}(\rho_{AB_{N-1}})
\end{eqnarray}
in Ref.\cite{XHLF} for all $0\leq\beta\leq\gamma,~\gamma\geq2$, where $\mathcal{V}_{p}=(\omega_p+l_p)^{\frac{\beta}{\gamma}}-l_p^{\frac{\beta}{\gamma}}$
for real numbers $\omega_{p}\geq 1$, $l_{p}\geq 1$ and $1\leq p\leq N-3,$ if $
\mathcal{C}_{a}^{\gamma}(\rho_{AB_i})\leq l_{i}\mathcal{C}_{a}^{\gamma}(\rho_{A|B_i+1\cdots B_N-1}),
\mathcal{C}_{a}^{\gamma}(\rho_{A|B_i\cdots B_N-1})\geq \mathcal{C}_{a}^{\gamma}(\rho_{AB_i})+\omega_{i}\mathcal{C}_{a}^{\gamma}(\rho_{A|B_i+1\cdots B_N-1}) $ for $i=1, 2, \cdots, z$, and
$l_{j}\mathcal{C}_{a}^{\gamma}(\rho_{AB_j})\geq \mathcal{C}_{a}^{\gamma}(\rho_{A|B_j+1\cdots B_N-1}),
\mathcal{C}_{a}^{\gamma}(\rho_{A|B_j\cdots B_N-1})\geq \omega_{j}\mathcal{C}_{a}^{\gamma}(\rho_{AB_j})+
\mathcal{C}_{a}^{\gamma}(\rho_{A|B_j+1\cdots B_N-1})$
for $j=z+1,\ldots,N-2,~1\leq z\leq N-3,~N\geq 4$.

Actually, serving as an indicator of the entanglement distribution among the subsystems, the monogamy inequalities adhered to by the CoA can be enhanced for increased tightness. We present the following lemma before giving our result.
\begin{lemma}\label{L2}
Consider $\frac{1}{2}\leq p\leq1$ and $0\leq r\leq\frac{1}{2}$ be arbitrary real numbers. For any $ 0\leq x\leq y,~0\leq y\leq1$ , it holds that
\begin{equation}\label{aaaa1}
\begin{aligned}
(1+x)^{r}-(px)^{r}\geq(1+y)^{r}-(py)^{r}.
\end{aligned}
\end{equation}

\begin{proof}
Let $u(x,r)=(1+x)^{r}-(px)^{r}$ for $ 0\leq x\leq y,~0\leq y\leq1$, $1/2\leq p\leq1$ and $0\leq r\leq1/2$. We have $\frac{\partial u(x,r)}{\partial x}=rx^{r-1}\big((1+\frac{1}{x})^{r-1}-p^{r}\big)\equiv rx^{r-1}v(x,r)$, where
$v(x,r)=(1+\frac{1}{x})^{r-1}-p^{r}$. Since
$\frac{\partial v(x,r)}{\partial r}=(1+\frac{1}{x})^{r-1}\ln(1+\frac{1}{x})-p^{r}\ln(p)$,
then $1+\frac{1}{x}\geq 2,~\ln(p)<0$, and  $v(x,r)$ increases with $r$, with $v_{\max}(x,r)=v(x,\frac{1}{2})=\big(\frac{1}{1+\frac{1}{x}}\big)^{\frac{1}{2}}
-p^{\frac{1}{2}}\leq0$. That is to say, $u(x,r)$ decreases with $x$.
As $0\leq x\leq y$, $u(x,r)\geq u(y,r)=(1+y)^{r}-(py)^{r}$. Therefore, we get the inequality (\ref{aaaa1}).
\end{proof}
\end{lemma}

Utilizing the above lemma, we derive below the general monogamy inequalities for the $\beta$-th ($0\leq\beta\leq\frac{\gamma}{2},~\gamma\geq2$) power of the CoA for multi-qudit GW states.

\begin{theorem}\label{coa2}
Let $\rho_{A_{j_{1}}\cdots A_{j_{t}}}$ denote the reduced density  matrix of the $\mathrm{GW}$ state $|\mathcal{W}_{n}^{d}\rangle_{A_{1}\cdots A_{n}}$ and $P=\{A,B,C\}$ a partition of the set
$\{A_{j_{1}}\cdots A_{j_{t}}\},~3\leq t\leq n$.
There exists real number $\omega\geq1$ such that\\
(1) if $\mathcal{C}_{aAB}^\gamma\leq l\mathcal{C}_{aAC}^\gamma$, the CoA  satisfies
\begin{eqnarray}\label{xx1}
\mathcal{C}_{aA|BC}^\beta\geq p^{\frac{\beta}{\gamma}}\mathcal{C}_{aAB}^\beta+[(\omega+l)^{\frac{\beta}{\gamma}}
-(pl)^{\frac{\beta}{\gamma}}]\mathcal{C}_{aAC}^\beta
\end{eqnarray}
for all $0\leq\beta\leq\frac{\gamma}{2},~\gamma\geq2,~1/2\leq p\leq1,~0\leq l \leq1$;\\
(2) if $\mathcal{C}_{aAC}^\gamma\leq l\mathcal{C}_{aAB}^\gamma$, the CoA  satisfies
\begin{eqnarray}\label{xx2}
\mathcal{C}_{aA|BC}^\beta\geq [(\omega+l)^{\frac{\beta}{\gamma}}-(pl)^{\frac{\beta}{\gamma}}]
\mathcal{C}_{aAB}^\beta+p^{\frac{\beta}{\gamma}}\mathcal{C}_{aAC}^\beta
\end{eqnarray}
for all $0\leq\beta\leq\frac{\gamma}{2},~\gamma\geq2,~1/2\leq p\leq1,~0\leq l \leq1$.
\end{theorem}

\begin{proof}
If $\mathcal{C}_{aAB}^\gamma\leq l\mathcal{C}_{aAC}^\gamma$,  we have
\begin{eqnarray}
\mathcal{C}_{aA|BC}^\beta&=&(\mathcal{C}_{aA|BC}^\gamma)^{\frac{\beta}{\gamma}}\geq(\mathcal{C}_{aAB}^\alpha+\omega \mathcal{C}_{aAC}^\gamma)^{\frac{\beta}{\gamma}}\nonumber\\
&=&\Big[\frac{\mathcal{C}_{aAB}^\gamma}{\omega \mathcal{C}_{aAC}^\gamma}+1\Big]^{\frac{\beta}{\gamma}}\omega ^{\frac{\beta}{\gamma}} \mathcal{C}_{aAC}^\beta
\nonumber\\
&\geq&\omega^{\frac{\beta}{\gamma}}\mathcal{C}_{aAC}^{\beta}\Big[\Big(1+\frac{l}{\omega}\Big)^{\frac{\beta}{\gamma}}
-\Big(p\frac{l}{\omega}\Big)^{\frac{\beta}{\gamma}}+\Big(p\frac{\mathcal{C}_{aAB}^\gamma}{\omega \mathcal{C}_{aAC}^\gamma}\Big)^{\frac{\beta}{\gamma}}
\Big]\nonumber\\
&=&p^\frac{\beta}{\gamma}\mathcal{C}_{aAB}^\beta+\big[(\omega+l)^{\frac{\beta}{\gamma}}-(pl)^{\frac{\beta}{\gamma}}\big]\mathcal{C}_{aAC}^\beta,
\end{eqnarray}
where the first inequality follows from (\ref{u}), and the second inequality from (\ref{aaaa1}) in Lemma \ref{L2}. Therefore, we get the inequality (\ref{xx1}). If $\mathcal{C}_{aAC}^\gamma\leq l\mathcal{C}_{aAB}^\gamma$, analogously we have $\mathcal{C}_{aA|BC}^\gamma\geq \mathcal{C}_{aAC}^\gamma+\omega \mathcal{C}_{aAB}^\gamma$, which gives rise to the inequality (\ref{xx2}).
\end{proof}

\noindent{\bf Remark 2} \, \
We have established a class of lower bounds for the $\beta$-th ($0\leq\beta\leq\frac{\gamma}{2}, ~\gamma\geq2$) power of CoA, utilizing the parameter $p$, where $1/2\leq p \leq1$. When $p=1$,  $\mathcal{C}_{aAB}^\gamma\leq l\mathcal{C}_{aAC}^\gamma$, the CoA satisfies
$$
\mathcal{C}_{aA|BC}^\beta\geq \mathcal{C}_{aAB}^\beta+[(\omega+l)^{\frac{\beta}{\gamma}}-l^{\frac{\beta}{\gamma}}]\mathcal{C}_{aAC}^\beta.
$$
Obviously our Theorem \ref{coa2} is a generalization of the inequality (\ref{XHLF2}) in Ref.~\cite{XHLF}. The lower bounds for the $\beta$-th power of the CoA become tighter as $p$ (or $l$) decreases, or $\omega$  increases, then
\begin{eqnarray}\label{bx}
\mathcal{C}_{aA|BC}^\beta&\geq& p^\frac{\beta}{\gamma}\mathcal{C}_{aAB}^\beta+\big[(\omega+l)^{\frac{\beta}{\gamma}}-(pl)^{\frac{\beta}{\gamma}}\big]\mathcal{C}_{aAC}^\beta\nonumber\\
&\geq&
\mathcal{C}_{aAB}^\beta+\big[(\omega+l)^{\frac{\beta}{\gamma}}-l^{\frac{\beta}{\gamma}}\big]\mathcal{C}_{aAC}^\beta\nonumber\\
&\geq& \mathcal{C}_{aAB}^\beta+\big((1+l)^\frac{\beta}{\gamma}
-l^\frac{\beta}{\gamma}\big)\mathcal{C}_{aAC}^\beta\\
&\geq& \mathcal{C}_{aAB}^\beta+(2^\frac{\beta}{\gamma}-1)\mathcal{C}_{aAC}^\beta\nonumber \label{XX3}
\end{eqnarray}
for all $0\leq\beta\leq\frac{\gamma}{2},~\gamma\geq2$, the second, third, and fourth equalities are satisfied when $p=1$, $\omega=1$ and $l=1$, respectively. Consequently, our results yield a tighter bound than that presented in equation (\ref{XHLF2}) of Ref.~\cite{XHLF}.
According to inequality (\ref{bx}), it follows that $(1+l)^\frac{\beta}{\gamma}
-l^\frac{\beta}{\gamma}=\frac{(1+k)^{\frac{\beta}{\gamma}}-1}{k^{\frac{\beta}{\gamma}}}$ if $l=\frac{1}{k}$ with $k\geq1$,
as well as $\frac{(1+k)^{\frac{\beta}{\gamma}}-1}{k^{\frac{\beta}{\gamma}}}\geq2^{\frac{\beta}{\gamma}}-1$ for $0\leq\beta\leq\gamma,~\gamma\geq2$. It is clear that the new established monogamy inequality for CoA is better than the inequality in Ref.~\cite{LZZ}. As a result, our inequality is notably tighter than the one detailed in Ref.~\cite{SX032344}.

\noindent{\bf Example 2} \, \
Let us consider the following three-qubit GW state,
\begin{equation}\label{Con66}
|\psi\rangle_{A|BC}=\frac{1}{\sqrt{6}}|100\rangle+\frac{1}{\sqrt{6}}|010\rangle+\frac{2}{\sqrt{6}}|001\rangle.
\end{equation}
We get $\mathcal{C}_{aA|BC}=\frac{\sqrt{5}}{3}$, $\mathcal{C}_{aAB}=\frac{1}{3}$ and
$\mathcal{C}_{aAC}=\frac{2}{3}$.
According to $l(\mathcal{C}_{aAC})^{3}\geq (\mathcal{C}_{aAB})^{3}$ and $\mathcal{C}_{aA|BC}^{3}\geq \mathcal{C}_{aAB}^{3}+\omega \mathcal{C}_{aAC}^{3}$, we have $\frac{1}{8}\leq l\leq1,~1\leq \omega\leq \frac{5\sqrt{5}-1}{8}$.  Therefore, when $N=3$ we have $\mathcal{C}_{aA|BC}^{\beta}=(\frac{\sqrt{5}}{3})^{\beta}$, $\mathcal{C}_{aAB}^\beta+(2^\frac{\beta}{\gamma}-1)\mathcal{C}_{aAC}^\beta=(\frac{1}{3})^\beta+(2^\frac{\beta}{\alpha}-1)(\frac{2}{3})^\beta$ from (\ref{SX}) in Ref.~\cite{SX032344},
$\mathcal{C}_{aAB}^\beta+(((1+k)^\frac{\beta}{\gamma}-1)/{k^\frac{\beta}{\gamma}})\mathcal{C}_{aAC}^\beta=(\frac{1}{3})^\beta+({(1+k)^\frac{\beta}{\alpha}-1})/{k^\frac{\beta}{\alpha}})(\frac{2}{3})^\beta$
 from (\ref{LYY}) in Ref.~\cite{LZZ},
$\mathcal{C}_{aAB}^\beta+[(\omega+l)^{\frac{\beta}{\gamma}}-l^{\frac{\beta}{\gamma}}]\mathcal{C}_{aAC}^\beta
=(\frac{1}{3})^\beta+[(\omega+l)^{\frac{\beta}{\gamma}}-l^{\frac{\beta}{\gamma}}](\frac{2}{3})^\beta$
 from (\ref{XHLF2}) in Ref.~\cite{XHLF} and $p^\frac{\beta}{\gamma}\mathcal{C}_{aAB}^\beta+\big[(\omega+l)^{\frac{\beta}{\gamma}}-(pl)^{\frac{\beta}{\gamma}}\big]\mathcal{C}_{aAC}^\beta
=p^\frac{\beta}{\gamma}(\frac{1}{3})^\beta+[(\omega+l)^{\frac{\beta}{\gamma}}-(pl)^{\frac{\beta}{\gamma}}](\frac{2}{3})^\beta$ from our result (\ref{xx1}).
Let $\omega=\frac{9}{8},~l=\frac{1}{k}=\frac{3}{4},~p=\frac{3}{4}~(p=\frac{1}{2})$.
It is evident that our result is better than the result (\ref{XHLF2}) in Ref.~\cite{XHLF},
thus also better than (\ref{LYY}) and (\ref{SX}) presented in Ref.~\cite{LZZ} and Ref.\cite{SX032344}, respectively, see Fig.\ref{Fig2}.
\begin{figure}
  \centering
  \includegraphics[width=10cm]{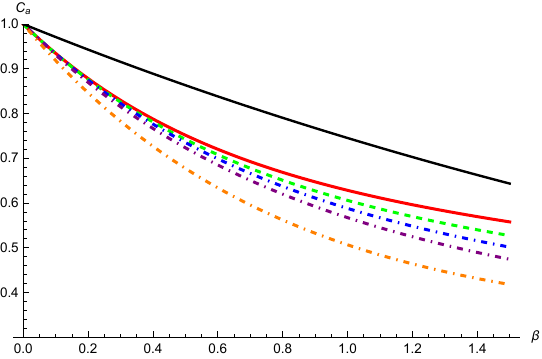}
  \caption{The $y$-axis represents the lower bound for the concurrence of assistance $\mathcal{C}^\beta_{aA|BC}$.
  Solid black line denotes $\mathcal{C}^\beta_{aA|BC}$ for the state presented in (\ref{Con66}). The red thick (green dashed) line represents the lower bound given in inequality (\ref{x1}) with $p=1/2$~($p=3/4$), and blue dot dashed thick represents the lower bound of (\ref{XHLF2}) from Ref.~\cite{XHLF}.
  The purple dot dashed and orange dot dashed line denotes the lower bound from Ref.~\cite{LZZ} and Ref.~\cite{SX032344}, respectively.}
  \label{Fig2}
\end{figure}

By repeatedly applying Theorem \ref{coa2}, we have extended the monogamy inequalities to the reduced density matrices of any $\mathrm{GW}$ state, thereby establishing the following  theorem.

\begin{theorem}\label{coa3}
Let $\rho_{A_{j_{1}}\cdots A_{j_{t}}}$ denote the reduced density  matrix of the GW state $|\mathcal{W}_{n}^{d}\rangle_{A_{1}\cdots A_{n}}$, and $P=\{A,B_{1},\cdots,B_{N-1}\}$ a partition of the set $\{A_{j_{1}}\cdots A_{j_{t}}\},~N\leq t\leq n$. Consider the real numbers $\omega_{r}$, $l_{r}$, $\delta_{r}$ and $p$, where $\omega_{r}\geq1$, $l_{r}\geq1$ , $\delta_{r}\geq1$, $1\le r\le N-2$ and $1/2\leq p\leq1$.
If $\mathcal{C}_{aAB_i}^\gamma\geq l_{i}^{\delta_{i}}\mathcal{C}_{aA|B_{i+1}\cdots B_{N-1}}^\gamma$, $\mathcal{C}_{aA|B_i\cdots B_{N-1}}^\gamma\geq \mathcal{C}_{aAB_i}^\gamma+\omega_i\mathcal{C}_{aA|B_{i+1}\cdots B_{N-1}}^\gamma$ for
$i=1,2,\cdots,z$, and if $\mathcal{C}_{aA|B_{j+1}\cdots B_{N-1}}^\gamma\geq l_{j}^{\delta_{j}}\mathcal{C}_{aAB_j}^\gamma$, $\mathcal{C}_{aA|B_j\cdots B_{N-1}}^\gamma\geq\omega_j\mathcal{C}_{aAB_j}^\gamma+\mathcal{C}_{aA|B_{j+1}\cdots B_{N-1}}^\gamma$ for $j=z+1,\cdots,N-2$, $1\leq z\leq N-3$ and $N\geq4$, we have
\begin{eqnarray}
\mathcal{C}_{aA|B_1\cdots B_{N-1}}^\beta&\geq&p^\frac{\beta}{\gamma}\mathcal{C}_{aAB_1}^\beta+
\sum\limits_{i=2}^{z}\prod\limits_{l=1}^{i-1}\Gamma_{l}p^\frac{(i-1)\beta}{\gamma}
\mathcal{C}_{aAB_{i}}^{\beta}\nonumber\\
& &+\Gamma_1\cdots\Gamma_{z}
\Big(\Gamma_{z+1}\mathcal{C}_{aAB_{z+1}}^{\beta}+\sum\limits_{j=z+2}^{N-2}
\Gamma_{j}p^\frac{(j-z-2)\beta}{\gamma}\mathcal{C}_{aAB_{j}}^{\beta}\Big)\nonumber\\
& &+\Gamma_1\cdots\Gamma_{z}p^\frac{(N-z-2)\beta}{\gamma}\mathcal{C}_{aAB_{N-1}}^\beta
\end{eqnarray}
for all $0\leq\beta\leq\frac{\gamma}{2},~\gamma\geq2$, where $\Gamma_{r}=(\omega_{r}+l_{r}^{\delta_{r}})^\frac{\beta}{\gamma}
-(l_{r}^{\delta_{r}})^\frac{\beta}{\gamma}$.
\end{theorem}

\begin{proof}
According to Theorem \ref{coa2}, if $\mathcal{C}_{aAB_i}^\gamma\geq l_{i}^{\delta_{i}}\mathcal{C}_{aA|B_{i+1}\cdots B_{N-1}}^\gamma$, $\mathcal{C}_{aA|B_i\cdots B_{N-1}}^\gamma\geq \mathcal{C}_{aAB_i}^\gamma+\omega_i\mathcal{C}_{aA|B_{i+1}\cdots B_{N-1}}^\gamma$,
$i=1,2,\cdots,z$, $\gamma\geq2$, we get
\begin{eqnarray}\label{x}
\mathcal{C}_{aA|B_1\cdots B_{N-1}}^\beta&\geq&p^\frac{\beta}{\gamma}\mathcal{C}_{aAB_1}^\beta+\Gamma_1\mathcal{C}_{aA|B_2\cdots B_{N-1}}^\beta\nonumber\\
&\geq&p^\frac{\beta}{\gamma}\mathcal{C}_{aAB_1}^\beta+\Gamma_1p^\frac{\beta}{\gamma}\mathcal{C}_{aAB_2}^\beta+\Gamma_1\Gamma_{2}\mathcal{C}_{aA|B_3\cdots B_{N-1}}^\beta\geq\cdots\nonumber\\
&\geq&p^\frac{\beta}{\gamma}\mathcal{C}_{aAB_1}^\beta+\Gamma_1p^\frac{\beta}{\gamma}\mathcal{C}_{aAB_2}^\beta+\cdots+\Gamma_1\cdots\Gamma_{z-1}p^\frac{\beta}{\gamma}\mathcal{C}_{aAB_z}^\beta\nonumber\\
& &+\Gamma_1\cdots\Gamma_{z}\mathcal{C}_{aA|B_{z+1}\cdots B_{N-1}}^\beta.
\end{eqnarray}
If $\mathcal{C}_{aA|B_{j+1}\cdots B_{N-1}}^\gamma\geq l_{j}^{\delta_{j}}\mathcal{C}_{aAB_j}^\gamma$, $\mathcal{C}_{aA|B_j\cdots B_{N-1}}^\gamma\geq\omega_j\mathcal{C}_{aAB_j}^\gamma+\mathcal{C}_{aA|B_{j+1}\cdots B_{N-1}}^\gamma$, $j=z+1,\cdots,N-2$, $\gamma\geq2$, we have
\begin{eqnarray}\label{y}
\mathcal{C}_{aA|B_{z+1}\cdots B_{N-1}}^\beta&\geq&\Gamma_{z+1}\mathcal{C}_{aAB_{z+1}}^\beta+p^\frac{\beta}{\gamma}\mathcal{C}_{aA|B_{z+2}\cdots B_{N-1}}^\beta\nonumber\\
&\geq&\Gamma_{z+1}\mathcal{C}_{aAB_{z+1}}^\beta+p^\frac{\beta}{\gamma}\Gamma_{z+2}\mathcal{C}_{aAB_{z+2}}^\beta+p^\frac{2\beta}{\gamma}\mathcal{C}_{aA|B_{z+3}\cdots B_{N-1}}^\beta\geq\cdots\nonumber\\
&\geq&\Gamma_{z+1}\mathcal{C}_{aAB_{z+1}}^\beta+p^\frac{\beta}{\gamma}\Gamma_{z+2}\mathcal{C}_{aAB_{z+2}}^\beta+\cdots\nonumber\\
& &+p^\frac{(N-z-3)\beta}{\gamma}\Gamma_{N-2}\mathcal{C}_{aAB_{N-2}}^\beta+p^\frac{(N-z-2)\beta}{\gamma}\mathcal{C}_{aAB_{N-1}}^\beta.
\end{eqnarray}
Through the conjunction of inequalities (\ref{x}) and (\ref{y}), we conclude the proof.
\end{proof}

\section{ Monogamy and polygamy of Tsallis-q entanglement of assistance in multi-qudit systems}\label{s4}
The Tsallis-$q$ entropy, which extends the concept of von Neumann entropy, holds significant importance within the realm of quantum information theory. It not only generalizes global quantum discord but also establishes a sufficient condition for the monogamy of an $N$-party quantum state \cite{PS}. For a pure state $|\psi\rangle_{AB}=\sum_i\sqrt{\lambda_i}|ii\rangle,$
the Tsallis-$q$ entanglement (TqEE) is defined as follows
\begin{equation}\label{T1}
\mathcal{T}_q(|\psi\rangle_{AB})=\frac{1-\mathrm{Tr}\rho_A^q}{q-1}=\frac{1-\sum_i \lambda_i^q}{q-1}
\end{equation}
for any $q>0$ and $q\ne 1$. The T$q$EE for a mixed state $\rho_{AB}$ is defined as presented in Ref.~\cite{KIM062328}.
\begin{equation*}
\mathcal{T}_q(\rho_{AB})=\min_{\{p_i,|\psi_{i}\rangle\}}\sum_i p_i\mathcal{T}_q(|\psi_{i}\rangle_{AB}),
\end{equation*}
where the minimum takes over all pure state decompositions of $\rho_{AB}$. As a dual concept of T$q$EE, the Tsallis-$q$ entanglement of assistance (T$q$EEoA) is defined as \cite{KIM062328}
\begin{equation*}
\mathcal{T}_q^a(\rho_{AB})=\max_{\{p_i,|\psi_{i}\rangle\}}\sum_i p_i\mathcal{T}_q(|\psi_{i}\rangle_{AB}),
\end{equation*}
where the maximum takes over all pure state decompositions of $\rho_{AB}.$

Combining the equalities $(\ref{C1})$ and $(\ref{T1})$, we get for $|\psi\rangle_{AB}=\sqrt{\lambda_0}|00\rangle+\sqrt{\lambda_1}|11\rangle,$ it follows that $\mathcal{C}^2(|\psi\rangle_{AB})=4\lambda_0\lambda_1$ and $\mathcal{T}_q(|\psi\rangle_{AB})=\frac{1-\lambda_0^q-\lambda_1^q}{q-1}$. Thus, $\mathcal{T}_q(|\psi\rangle_{AB})=f_q(\mathcal{C}^2(|\psi\rangle_{AB}))$, where
\begin{equation*}
 f_q(\theta)=\frac{1}{q-1}[1-\big(\frac{1+\sqrt{1-\theta}}{2})^q-(\frac{1-\sqrt{1-\theta}}{2}\big)^q].
\end{equation*}
The TqEE can be considered as a generalization of the entanglement of formation $E_{f}$ as $q$ tends to 1, which is defined by \cite{BBPS,BDSW},
$E_{f}(\rho_{AB})=\mathrm{min}\sum_{i}p_{i}E_{f}(|\psi_{i}\rangle_{AB})$,
where $E_{f}(|\psi_{i}\rangle_{AB})=-\mathrm{Tr}\rho_{A}^{i}\ln\rho_{A}^{i}=-\mathrm{Tr}\rho_{B}^{i}\ln\rho_{B}^{i}$
is the von Neumann entropy of $|\psi_{i}\rangle_{AB}$, where the minimum takes over all pure state decompositions of $\rho_{AB}$.

Recently, two new relations concerning multi-partite quantum entanglement have been introduced in T$q$EE~\cite{SX032344}. Let $\rho_{A_{j_{1}}\cdots A_{j_{t}}}$ denote the reduced density  matrix of the GW state $|\mathcal{W}_{n}^{d}\rangle_{A_{1}\cdots A_{n}}$, and $\{Q_1,Q_2,\cdots,Q_s\}$ a partition of the set $\{A_{j_1},A_{j_2},\cdots,A_{j_t}\}$, $s\leq t<n$. When $q\in[\frac{5-\sqrt{13}}{2},\frac{5+\sqrt{13}}{2}],$ for $\alpha\ge 2$ one has
\begin{equation}\label{SX1}
	\mathcal{T}_q^{\alpha}(\rho_{Q_1|Q_2\cdots Q_s})\ge \sum_{i=2}^{s} \mathcal{T}_q^{\alpha}(\rho_{Q_1Q_i}).
\end{equation}
When $q\in[\frac{5-\sqrt{13}}{2},2]\cup[3,\frac{5+\sqrt{13}}{2}]$, for $0\leq\beta\le 1$ one has
\begin{equation}\label{SX2}
	\mathcal{T}_q^{\beta}(\rho_{Q_1|Q_2\cdots Q_s})\le \sum_{i=2}^{s} \mathcal{T}_q^{\beta}(\rho_{Q_1Q_i}).
\end{equation}
Analytical formulae of T$q$EE and T$q$EEoA of a
reduced density matrix for a GW state has been also obtained ~\cite{SX032344},
\begin{equation}\label{ddd}
\mathcal{T}_q^a(\rho_{A_{j_1}|A_{j_2}\cdots A_{j_t}})=\mathcal{T}_q(\rho_{A_{j_1}|A_{j_2}\cdots A_{j_t}})=f_q(\mathcal{C}^2(\rho_{A_{j_1}|A_{j_2}\cdots A_{j_t}})),
\end{equation}
where $q\in[\frac{5-\sqrt{13}}{2},2]\cup[3,\frac{5+\sqrt{13}}{2}].$

Concerning the monogamy and polygamy relation about T$q$EEoA for the GW states, we have the following theorems.

\begin{theorem}
Let $\rho_{A_{j_{1}}\cdots A_{j_{t}}}$ denote the reduced density  matrix of the GW state $|\mathcal{W}_{n}^{d}\rangle_{A_{1}\cdots A_{n}}$, and $\{Q_1,Q_2,\cdots,Q_s\}$ a partition of the set $\{A_{j_1},A_{j_2},\cdots,A_{j_t}\}$, $s\leq t<n$. For $q\in[\frac{5-\sqrt{13}}{2},2]\cup[3,\frac{5+\sqrt{13}}{2}]$, we see that
\begin{eqnarray}\label{TQ2}
(\mathcal{T}_q^a)^\alpha(\rho_{Q_1|Q_2\cdots Q_s})\geq\sum_{i=2}^{s} (\mathcal{T}_q^a)^\alpha(\rho_{Q_1Q_i})
\end{eqnarray}
for $\alpha\geq2$.
\end{theorem}

\begin{proof}
According to the relation (\ref{ddd}), when $\alpha\geq2$, $q\in[\frac{5-\sqrt{13}}{2},2]\cup[3,\frac{5+\sqrt{13}}{2}],$ we have
\begin{eqnarray}\nonumber
(\mathcal{T}_q^a)^\alpha(\rho_{Q_1|Q_2\cdots Q_s})
&=&(\mathcal{T}_q)^\alpha(\rho_{Q_1|Q_2\cdots Q_s})\\[1mm]\nonumber
&\geq&\sum_{i=2}^{s} (\mathcal{T}_q)^\alpha(\rho_{Q_1Q_i})\\[1mm]\nonumber
&=& \sum_{i=2}^{s} (\mathcal{T}_q^a)^\alpha(\rho_{Q_1Q_i}),
\end{eqnarray}
where the inequality follows from (\ref{SX1}).
\end{proof}

\noindent{\bf Example 3} \, \ Consider the three-qubit GW state given in (\ref{Con66}) again. By using equality (\ref{ddd}), we get the T$q$EEoA of $|\psi\rangle_{ABC}$. The residual quantity of T$q$EEoA is given by $(\mathcal{T}_q^a)^\alpha(\rho_{A|BC})-(\mathcal{T}_q^a)^\alpha(\rho_{AB})-(\mathcal{T}_q^a)^\alpha(\rho_{AC})$. We plot the residual quantity as a function of $q$ and $\alpha$. Fig.\ref{Fig3} and Fig.\ref{Fig4} show that the residual quantity is always positive for $q\in[\frac{5-\sqrt{13}}{2},2]\cup[3,\frac{5+\sqrt{13}}{2}]$ and $\alpha\geq2$.
 \begin{figure}[htbp]
\centering
\begin{minipage}[t]{0.48\textwidth}
\centering
\includegraphics[width=7cm]{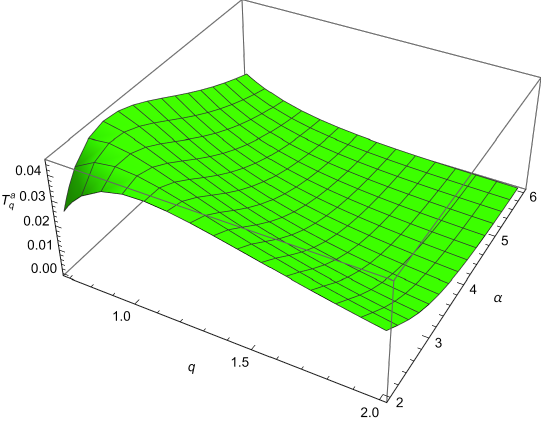}
\caption{\scriptsize The green  surface represents the residual quantity with $q\in[\frac{5-\sqrt{13}}{2},2]$.}
\label{Fig3}
\end{minipage}
\begin{minipage}[t]{0.48\textwidth}
\centering
\includegraphics[width=7cm]{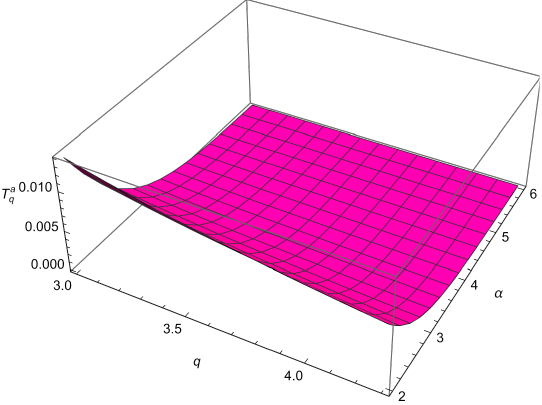}
\caption{\scriptsize The red surface represents the residual quantity with $q\in[3,\frac{5+\sqrt{13}}{2}]$.}
\label{Fig4}
\end{minipage}
\end{figure}

Similarly, we derive the following new polygamy relation about T$q$EEoA
for the GW states when $0\leq\beta\leq1$.

\begin{theorem}
Let $\rho_{A_{j_1}A_{j_2}\cdots A_{j_t}}$ denote the reduced density matrix of a GW state $|\mathcal{W}_{n}^{d}\rangle_{A_{1}\cdots A_{n}}$ and $\{Q_1,Q_2,\cdots,Q_s\}$ a partition of the set $\{A_{j_1},A_{j_2},\cdots,A_{j_t}\}$, $s\leq t<n$. For $q\in[\frac{5-\sqrt{13}}{2},2]\cup[3,\frac{5+\sqrt{13}}{2}]$, we see that
\begin{eqnarray}\label{TQ3}
(\mathcal{T}_q^a)^\beta(\rho_{Q_1|Q_2\cdots Q_s})\leq\sum_{i=2}^{s} (\mathcal{T}_q^a)^\beta(\rho_{Q_1Q_i}),
\end{eqnarray}
where $0\leq\beta\leq1$.
\end{theorem}

\section{conclusion and discussion}\label{s5}
The concepts of monogamy and polygamy relations highlight the intrinsic characteristics of multi-qudit entanglement, showcasing the complex nature of quantum entanglement. Initially, we elucidated the expressions of multi-qudit monogamy constraints through the CoA, confirming that these new inequalities delineate more detailed entanglement distributions compared to the previous ones. Furthermore, using T$q$EEoA with $\frac{5-\sqrt{13}}{2}\leq q\leq 2$ and $3\leq q\leq\frac{5+\sqrt{13}}{2}$, we have formulated new monogamy and polygamy relations for multi-qudit GW states. Remarkably, our inequality (\ref{TQ3}) is particularly significant as it is not only valid for the discussed systems but also for specific high-dimensional quantum systems that defy the CKW monogamy inequality (\ref{CKW}).

Consider a higher-dimensional example to demonstrate the polygamy inequality (\ref{TQ3}). Kim et al. demonstrate that the $\mathrm{CKW}$ inequality does not apply to the pure $3\otimes2\otimes2$ state $|\psi_{ABC}\rangle$,
$|\psi_{ABC}\rangle=\frac{1}{6}(\sqrt2|121\rangle+\sqrt2|212\rangle+|311\rangle+|322\rangle)$ \cite{KIM012329}. The T$q$EEoA is given by $T_{q}^{a}(|\psi_{A|BC}\rangle)=\frac{1}{q-1}[1-(\frac{1}{3})^{q-1}]$. The bipartite reduced state of subsystem $AB$ is expressed as
$\rho_{AB}=\frac{1}{2}(|\mu\rangle_{AB}\langle \mu|+|\nu\rangle_{AB}\langle \nu|)$,
with
$|\mu\rangle_{AB}=\frac{\sqrt2}{\sqrt3}|12\rangle+\frac{1}{\sqrt3}|31\rangle$ and
$|\nu\rangle_{AB}=\frac{\sqrt2}{\sqrt3}|21\rangle+\frac{1}{\sqrt3}|32\rangle$.
It can be demonstrated that for any arbitrary pure states $|\phi_{AB}\rangle=c_\mu|\mu\rangle_{AB}+c_\nu|\nu\rangle_{AB}$ with $|c_\mu|^2+|c_\nu|^2=1$, their reduced state $\rho_A=\mathrm{Tr}_B(|\phi\rangle_{AB}\langle\phi|)$ has the same spectrum $\{0,1/3,2/3\}$. Then, the T$q$EE of $|\phi_{AB}\rangle$ is $T_{q}(|\phi_{AB}\rangle)=\frac{1}{q-1}[1-(1+2^q)(\frac{1}{3})^{q}].$ Thus, the T$q$EE of $\rho_{AB}$ is $T_{q}(\rho_{AB})=\frac{1}{q-1}[1-(1+2^q)(\frac{1}{3})^{q}].$ In the same way, the T$q$EE of $\rho_{AC}$ is $T_{q}(\rho_{AC})=\frac{1}{q-1}[1-(1+2^q)(\frac{1}{3})^{q}].$  From the relation  (\ref{TQ3}) of T$q$EEoA, we get
\begin{eqnarray}
T_q^a(|\psi_{A|BC}\rangle)-T_q^a(\rho_{AB})-T_q^a(\rho_{AC})
&\leq& T_q^a(|\psi_{A|BC}\rangle)-T_q(\rho_{AB})-T_q(\rho_{AC})\nonumber\\
&=&\frac{1}{q-1}\Big\{[1-(\frac{1}{3})^{q-1}]-2[1-(1+2^q)(\frac{1}{3})^{q}]\Big\}.
\end{eqnarray}
As shown in Fig.\ref{Fig5}, the polygamy inequality (\ref{TQ3}) still works in this multipartite higher-dimensional system. In other words, the relation  (\ref{TQ3})  based on T$q$EEoA remains valid for high-dimensional quantum systems, which defy the CKW monogamy inequality(\ref{CKW}). Our findings could provide valuable insights for deeper investigations into the distribution of entanglement within higher-dimensional systems.

\begin{figure}\label{TQ}
  \centering
  \includegraphics[width=10cm]{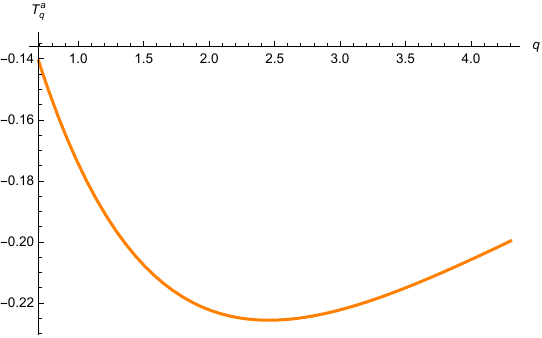}
  \caption{Inequality (\ref{TQ3}) holds for $q\in[\frac{5-\sqrt{13}}{2},2]\cup[3,\frac{5+\sqrt{13}}{2}]$ (solid orange line).
  }
\label{Fig5}
\end{figure}

\bigskip

\noindent{\bf Acknowledgments}
This work is supported in part by the National Natural Science Foundation of China under Grant Nos. 12075159 and 12171044, the specific research fund of the Innovation Platform for Academicians of Hainan Province.

\bigskip
\noindent{\bf Data availability statement}\
No new data were created or analyzed in this study.

\end{document}